\documentclass[11pt]{article}
\usepackage{amssymb,amsmath,amsthm,amscd,latexsym}
\usepackage{mathrsfs}
\usepackage{mathrsfs}
\usepackage{amsfonts}
\usepackage{amsmath}
\usepackage{amssymb}
\usepackage{amscd}
 \usepackage{color}
 \usepackage{bm}

\renewcommand{\paragraph}{\roman{paragraph}}
\input cyracc.def

\newcommand{\F}{\mathbb{F}}

\newtheorem{thm}{\bfseries  Theorem}[section]
\newtheorem{lem}[thm]{\bfseries   Lemma}

\begin{document}
\title{\bf New constructions of asymptotically optimal codebooks via character sums over a local ring
\thanks{This research is supported by the National Natural Science Foundation of China
under Grant 11771007 and Grant 61572027.}
}
\author{Liqin Qian\thanks{Department of Mathematics, Nanjing University of Aeronautics and Astronautics, Nanjing, Jiangsu, 210007, China, {\tt qianliqin\_1108@163.com}},
Xiwang Cao\thanks{Corresponding author, Department of Mathematics, Nanjing University of Aeronautics and Astronautics, Nanjing, Jiangsu, 210007, China; Laboratory of Information Security, Institute of Information Engineering, Chinese Academy of Sciences, Beijing 100042, China, {\tt xwcao@nuaa.edu.cn}},
Wei Lu\thanks{School of Mathematics, Southeast University, Nanjing, Jiangsu, 211189, China, {\tt  luwei1010@seu.edu.cn}},
Xia Wu\thanks{School of Mathematics, Southeast University, Nanjing, Jiangsu, 211189, China, {\tt wuxia80@seu.edu.cn}}
}

\date{}
\maketitle
\begin{abstract}
In this paper, we present explicit description on the additive characters, multiplicative characters and Gauss sums over a local ring. As an
application, based on the additive characters and multiplicative characters satisfying certain conditions, two new constructions of complex codebooks over a local ring are introduced. With these two constructions, we obtain two families of codebooks achieving asymptotically optimal with respect to the Welch bound. It's worth mentioning that the codebooks constructed in this paper have new parameters.
\end{abstract}

\section{Introduction}
Let $C=\{\textbf{c}_0,\textbf{c}_1,\cdots, \textbf{c}_{N-1}\}$ be a set of $N$ unit-norm complex vectors $\textbf{c}_l\in \mathbb{C}^K$ over an alphabet $A$,  where $l=0, 1,\cdots , N-1$. The size of $A$ is called the alphabet size of $C$. Such a set $C$ is called an $(N, K)$ codebook (also called a signal set). The maximum cross-correlation amplitude, which is a performance measure of a codebook in practical applications, of the $(N, K)$ codebook $C$ is defined as
\begin{eqnarray*}
  I_{\max}(C) &=& \max_{0\leq i<j\leq N-1} |\textbf{c}_i\textbf{c}_j^H|,
\end{eqnarray*}
where $\textbf{c}_j^H$ denotes the conjugate transpose of the complex vector $\textbf{c}_j$. For a certain length $K$, it is desirable to design a codebook such that the number $N$ of codewords is as large as possible and the maximum cross-correlation amplitude $I_{\max}(C)$ is as small as possible. To evaluate a codebook $C$ with parameters $(N, K)$, it is important to find the minimum achievable $I_{\max}(C)$ or its lower bound. However, for $I_{\max}(C)$, we have the well-known Welch bound in the following.
\begin{lem}\label{lem1} \cite{LW} For any $(N, K)$ codebook $C$ with $N\geq K$,
\begin{eqnarray}\label{eq1}
  I_{\max}(C) &\geq& I_w=\sqrt{\frac{N-K}{(N-1)K}}.
\end{eqnarray}
Furthermore, the equality in (\ref{eq1}) is achieved if and only if
\begin{eqnarray*}
  |\textbf{c}_i\textbf{c}_j^H| &=& \sqrt{\frac{N-K}{(N-1)K}}
\end{eqnarray*}
for all pairs $(i, j)$ with $i\neq j$.
\end{lem}
A codebook is referred to as a maximum-Welch-bound-equality (MWBE) codebook \cite{DS} or an equiangular tight frame \cite{JK} if it meets the Welch bound equality in ($\ref{eq1}$). Codebooks meeting the Welch bound are used to distinguish among the signals of different users in code-division multiple-access (CDMA) systems \cite{MM}. In addition, codebooks meeting optimal (or asymptotically optimal) with respect to the Welch bound are much preferred in many practical applications, such as, multiple description coding over erasure channels \cite{SH}, communications \cite{DS}, compressed sensing \cite{CW}, space-time codes \cite{TK}, coding theory \cite{DGS} and quantum computing \cite{RBSC} etc.. In general, it is very difficult to construct optimal codebooks achieving the Welch bound (i.e. MWBE). Hence, many researchers attempted to construct asymptotically optimal codebooks, i.e., the minimum achievable $I_{\max}(C)$ of the codebook $C$ nearly achieving the Welch bound for large $N$. There are many results in regard to optimal or almost optimal codebooks by the Welch bound, interested readers may refer to [1, 2, 4-6, 9-12, 14-16, 18, 28, 29]. It is important that the construction method of codebooks. At present, many researchers constructed the codebooks based on difference sets, almost difference sets, relative difference sets, binary row selection sequences and cyclotomic classes.

It is well known that the additive characters, multiplicative characters and Gauss sums over finite fields and some of their good properties
\cite[Chapter 5]{LNC}. Especially, they have many rich applications in coding theory. It's worth mentioning that some researchers constructed codebooks by using the character sums of finite fields \cite{DF, LC, ZF}. Later, G. Luo and X. Cao proposed two constructions of complex codebooks from character sums over the Galois ring $GR(p^2, r)$ in \cite{LC2} based on existing results \cite{LZF}. In fact, we know that many scholars have done a lot of research over local rings \cite{GSF, LL, SWLS} etc..  Motivated by \cite{LC2} and \cite{LZF}, a natural question is to explore the character sums over the ring $R=\F_q+u\F_q~(u^2=0)$, is it possible to construct codebooks over the ring $R$ based on the character sums we studied and obtain several classes of asymptotically optimal codebooks with respect to the Welch bound?

This paper will give a positive answer to this question. This manuscript has three main contributions. One contribution of this paper is to give explicit description on the additive characters, multiplicative characters and establish a Gauss sum over a local ring for the first time. Another contribution of this paper is to focus on the constructions of codebooks over the ring $R=\F_q+u\F_q~(u^2=0)$ by using the character sums. Finally, we show that the maximum cross-correlation amplitudes $I_{\max}(C)$ of these codebooks asymptotically meet the Welch bound and obtain new parameters by comparing with the parameters of some known classes of asymptotically optimal codebooks.

The rest of this paper is arranged as follows. Section 2 presents some notations and basic results which will be
needed in subsequent sections. In Section 3, we explicit description on the additive characters and multiplicative characters over a local ring. In Section 4, we present computation on Gauss sums over a local ring. Section 5 introduces two generic families of codebooks meeting asymptotically optimal with respect to the Welch bound. In Section 6, we conclude this paper and present several open problems.

\section{Preliminaries}
Let $\F_q$ denote the finite field with $q$ elements and $q=p^m$, where $p$ is a prime and $m$ is a positive integer. We consider the chain ring $R=\F_q+u\F_q=\{a+bu: a, b\in \F_q\} (u^2=0)$ with the unique maximal ideal $M=\langle u\rangle$. In fact, $R=\F_q\oplus u\F_q \simeq \F_q^2$ is a two-dimensional vector space over $\F_q$ and $|R|=q^2.$ The invertible elements of $R$ is $$R^*=R\backslash M=\F_q^*+u\F_q=\{a+bu: a\in \F_q^*, b\in \F_q\}$$ with $|R^*|=q(q-1)$. In fact, $R^*$ can also be represented as $\F_q^*\times (1+M)~~({\rm direct~product}).$

A character $\chi$ of a finite abelian group $G$ is a homomorphism from $G$ into the multiplicative group $U$ of complex numbers of absolute value 1, that is, a mapping from $G$ into $U$ with $\chi(g_1g_2)=\chi(g_1)\chi(g_2)$ for all $g_1, g_2\in G.$ Next, we recall the the additive characters and multiplicative characters of the finite field $\F_q$.

$\bullet$ The additive character $\chi$ of $\F_q$ defined by $$\chi(c)=e^{\frac{2\pi i{\rm Tr}(c)}{p}}$$ for all $c\in \F_q$, where Tr: $\F_q\longrightarrow \F_p$ is the absolute trace function from $\F_q$ to $\F_p$~(see Definition 2.22 in \cite{LNC}). For any $c_1, c_2\in \F_q$, we have
\begin{eqnarray}\label{den2}
  \chi(c_1+c_2) &=&\chi(c_1)\chi(c_2).
\end{eqnarray}
Moreover, for $b\in \F_q$, the function $\chi_b$ is defined as $\chi_b(c)=\chi(bc)$ for all $c\in \F_q$.\\
$\bullet$ The multiplicative character $\psi_j$ of $\F_q$ defined by $$\psi_j(g^k)=e^{\frac{2\pi ijk}{q-1}}$$ for each $j=0,1,\cdots, q-2$, where $k=0,1,\cdots, q-2$ and $g$ is a fixed primitive element of $\F_q$. For any $c_1, c_2\in \F_q^*$, we have
\begin{eqnarray}\label{den3}
  \psi_j(c_1c_2) &=&\psi_j(c_1)\psi_j(c_2).
\end{eqnarray}
Now, let $\psi$ be a multiplicative and $\chi$ an additive character of $\F_q$. Then the Gauss sum $G(\psi, \chi)$ of $\F_q$ is defined by
\begin{eqnarray*}
  G(\psi, \chi)&=&\sum\limits_{c\in \F_q^*}\psi(c)\chi(c).
\end{eqnarray*}
However, we now need to study the additive and multiplicative characters of a local ring $R=\F_q+u\F_q~(u^2=0)$, which implies that the character of the ring $R$ are described in detail similarly by the definition of the character of the finite field $\F_q$. Furthermore, the explicit description on the additive and multiplicative characters of $R$ we present should be satisfied the similar properties above equalities (\ref{den2}) and (\ref{den3}), respectively. In addition, we establish the Gauss sum of $R$ by the Gauss sum of $\F_q$. Hence, we will present the the additive and multiplicative characters of $R$ in the following section based on the characters of finite fields and propose the Gauss sum of $R$ in Section 4.
\section{Characters}
In this section, we will give the additive characters and multiplicative characters of $R$.\\\\
$\blacktriangle$\textbf{ \large Additive characters of $R$}

The group of additive characters of $(R, +)$ is
$$\widehat{R}:=\{\lambda: R\longrightarrow \mathbb{C}^*| \lambda(\alpha+\beta)=\lambda(\alpha)\lambda(\beta), \alpha, \beta \in R\}.$$
For any additive character $\lambda$ of $R$ $$\lambda: R \longrightarrow\mathbb{C}^*.$$ Since $\lambda(a_0+ua_1)=\lambda(a_0)\lambda(ua_1)$ for any $a_0, a_1\in \F_q,$ we define two maps as follows:
\begin{itemize}
  \item $$\lambda^{'}: \F_q \longrightarrow \mathbb{C}^*$$ by $\lambda^{'}(c):=\lambda(c)$ for $c\in \F_q.$
  \item $$\lambda^{''}: \F_q \longrightarrow \mathbb{C}^*$$ by $\lambda^{''}(c):=\lambda(uc)$ for $c\in \F_q.$
\end{itemize}
Therefore, it is easy to prove that $\lambda^{'}(c_1+c_2)=\lambda^{'}(c_1)\lambda^{'}(c_2)$ and $\lambda^{''}(c_1+c_2)=\lambda^{''}(c_1)\lambda^{''}(c_2)$ for $c_1, c_2 \in \F_q.$ Based on this, we know that $\lambda^{'}$ and $\lambda^{''}$ are additive characters of $(\F_q, +)$, then there exist $b, c \in \F_q$ such that $$\lambda^{'}(x)=\zeta_p^{{\rm Tr}(bx)}=\chi_{b}(x), \lambda^{''}(x)=\zeta_p^{{\rm Tr}(cx)}=\chi_{c}(x)$$ for all $x\in \F_q$, where $\zeta_p=e^{\frac{2\pi i}{p}}$ is a primitive $p$th root of unity over $\F_q.$ Hence, we can get the additive character of $R$
\begin{eqnarray*}
  \lambda(a_0+ua_1) &=& \lambda'(a_0)\lambda''(a_1)\\
  &=& \chi_{b}(a_0)\chi_{c}(a_1).
\end{eqnarray*}
Thus, there is an one-to-one correspondence:
\begin{eqnarray*}
 \tau : \widehat{(R,+)} &\longrightarrow& \widehat{(\mathbb{F}_q,+)}\times \widehat{(\mathbb{F}_q,+)},\\
  \lambda &\longmapsto& (\chi_b, \chi_c).
\end{eqnarray*}
It is easy to prove that the mapping $\tau$ is an isomorphism.\\
$\blacktriangle$\textbf{ \large Multiplicative characters of $R$}

The structure of the multiplicative group $R^*$ is $$R^*=\F_q^*\times (1+M)~~({\rm direct~product}).$$ Now, we have
\begin{eqnarray*}
  R^* &=& \{a_0+ua_1: a_0\in \F_q^*, a_1\in \F_q\} \\
   &=& \{b_0(1+ub_1): b_0\in \F_q^*, b_1\in \F_q\}.
\end{eqnarray*}

The group of multiplicative characters of $R$ is denoted by $\widehat{R}^*$ and $\widehat{R}^*=\widehat{\F}_q^*\times\widehat{(1+M)}$. We define $$\widehat{R}^*:=\{\varphi: R^*\longrightarrow \mathbb{C}^*| \varphi(\alpha\beta)=\varphi(\alpha)\varphi(\beta), \alpha, \beta \in R\}.$$
For any multiplicative character $\varphi$ of $R$ $$\varphi: R^* \longrightarrow\mathbb{C}^*.$$ Since $\varphi(b_0(1+ub_1))=\varphi(b_0)\varphi(1+ub_1)$ for any $b_0\in \F_q^*, b_1\in \F_q,$ we define two maps as follows:
\begin{itemize}
  \item $$\varphi^{'}: \F_q^*\longrightarrow \mathbb{C}^*$$ by $\varphi^{'}(c):=\varphi(c)$ for $c\in \F_q^*.$
  \item $$\varphi^{''}: \F_q\longrightarrow \mathbb{C}^*$$ by $\varphi^{''}(c):=\varphi(1+uc)$ for $c\in \F_q.$
\end{itemize}
For any $c_1, c_2 \in \F_q^*$, we have $\varphi'(c_1c_2)=\varphi'(c_1)\varphi'(c_2)$ and
\begin{eqnarray*}
  \varphi''(c_1+c_2)&=& \varphi(1+u(c_1+c_2)) \\
   &=& \varphi((1+uc_1)(1+uc_2)) \\
   &=& \varphi(1+uc_1)\varphi(1+uc_2)\\
   &=&\varphi''(c_1)\varphi''(c_2).
\end{eqnarray*}
Based on this, we can obtain that $\varphi'$ is a multiplicative character of $\F_q$ and $\varphi''$ is an additive character of $\F_q$. Hence, we can get the multiplicative character of $R$ $$\varphi(b_0(1+ub_1))=\varphi'(b_0)\varphi''(b_1),$$ where $\varphi'\in \widehat{\F}_q^*$ and $\varphi''\in \widehat{\F}_q.$ Since $\varphi''$ is an additive character of $\F_q$, then there exists $a\in \F_q$ such that $\varphi''=\chi_a.$
Moreover, we have
\begin{eqnarray*}
\sigma :  \widehat{(R^*,\ast)} &\longrightarrow& \widehat{\mathbb({\F}_q^*,\ast)}\times \widehat{(\mathbb{F}_q,+)}, \\
  \varphi &\longmapsto& (\psi, \chi_a),
\end{eqnarray*}
where $\psi=\varphi'$ is a multiplicative character of $\F_q$. One can show that the mapping $\sigma$ is an isomorphism.
\section{Gaussian sums}
Let $\lambda$ and $\varphi$ be an additive character and a multiplicative character of $R$, respectively. The Gaussian sum
for $\lambda$ and $\varphi$ of $R=\F_q+u\F_q~(u^2=0)$ is defined by
\begin{eqnarray*}
  G_R(\varphi, \lambda) &=&\sum\limits_{t\in R^*}\varphi(t)\lambda(t).
\end{eqnarray*}

In this section, we calculate the value of $G_R(\varphi, \lambda)$. For convenience, we denote $\varphi:=\psi\star\chi_a~({\rm namely,}~\varphi(t)=\psi(t_0)\chi_a(t_1)), \lambda:=\chi_b\star\chi_c~({\rm namely,}~\lambda(t)=\chi_b(t_0)\chi_c(t_0t_1))$ according to Section 3, where $a, b, c\in \F_q$ and $t=t_0(1+ut_1)\in R.$ Hence, we denote $G_R(\varphi, \lambda):=G(\psi\star\chi_a, \chi_b\star\chi_c)$.
\begin{thm}\label{thm1}
 Let $\varphi$ be a multiplicative character and $\lambda$ be an additive character of $R$, where $\varphi:=\psi\star\chi_a, \lambda:=\chi_b\star\chi_c$ and $a, b, c\in \F_q.$ Then the Gaussian sum $G_R(\varphi, \lambda)$ satisfies
 \begin{equation*}\label{den1}
G_R(\varphi, \lambda)=\begin{cases}
\emph{ }qG(\psi, \chi_b),  ~~~~~~~~~{\rm if}~a=0, c=0;\\
   \emph{ }0, ~~~~~~~~~~~~~~~~~~~{\rm if}~a=0, c\neq0; \\
   \emph{ }0, ~~~~~~~~~~~~~~~~~~~{\rm if}~a\neq0, c=0; \\
   \emph{ }q\psi(-\frac{a}{c})\chi(-\frac{ab}{c}), ~~~{\rm if}~a\neq0, c\neq 0, \\
\end{cases}
\end{equation*}
where
\begin{equation*}
G(\psi, \chi_b)=\begin{cases}
\emph{ }q-1,  ~~~~{\rm if}~\psi~{\rm is~trivial}, b=0;\\
   \emph{ }-1, ~~~~~{\rm if}~\psi~{\rm is~trivial}, b\neq0; \\
   \emph{ }0,~~~~~~~~~{\rm if}~\psi~{\rm is~nontrivial}, b=0.\\
\end{cases}
\end{equation*}If $\psi$ is nontrivial and $b\neq0$, then $|G(\psi, \chi_b)|=q^\frac{1}{2}$.
\end{thm}
\begin{proof} Now, let $\varphi:=\psi\star\chi_a$ and $\lambda:=\chi_b\star\chi_c$ with $a,b,c\in \F_q.$ Assume that $t=t_0(1+ut_1)$, where $t_0\in \F_q^*$ and $t_1\in \F_q.$
\begin{eqnarray*}
  G_R(\varphi, \lambda) &=& \sum\limits_{t\in R^*}\varphi(t)\lambda(t)\\
   &=& \sum\limits_{t_0\in \F_q^*, t_1\in \F_q}\varphi(t_0(1+ut_1))\lambda(t_0(1+ut_1)) \\
   &=&  \sum\limits_{t_0\in \F_q^*, t_1\in \F_q}\psi(t_0)\chi_a(t_1)\chi_b(t_0)\chi_c(t_0t_1)\\
   &=&  \sum\limits_{t_0\in \F_q^*, t_1\in \F_q}\psi(t_0)\chi(at_1+bt_0+ct_0t_1)\\
    &=&  \sum\limits_{t_0\in \F_q^*}\psi(t_0)\chi(bt_0)\sum\limits_{t_1\in \F_q}\chi((a+ct_0)t_1)\\
    &=&  q\sum\limits_{t_0\in \F_q^*, a+ct_0=0}\psi(t_0)\chi(bt_0)\\
     &=&\begin{cases}
 \emph{ }qG(\psi, \chi_b),  ~~~~~~~~~{\rm if}~a=0, c=0;\\
   \emph{ }0, ~~~~~~~~~~~~~~~~~~~~{\rm if}~a=0, c\neq0; \\
   \emph{ }0, ~~~~~~~~~~~~~~~~~~~~{\rm if}~a\neq0, c=0; \\
   \emph{ }q\psi(-\frac{a}{c})\chi(-\frac{ab}{c}),~~{\rm if}~a\neq0, c\neq 0, \\
\end{cases}
\end{eqnarray*}
where $G(\psi, \chi_b)$ is a Gaussian sum of $\F_q.$
\end{proof}

\section{Two families of asymptotically optimal codebooks}
In this section, we study two classes of codebooks asymptotically
achieving the Welch bound by using character sums over the local ring $R=\F_q+u\F_q~(u^2=0)$. Note that $|R^*|=q(q-1)$ and we can write $K=q(q-1)$. Let $\varphi:=\psi\star\chi_a$ and $\lambda:=\chi_b\star\chi_c$ with $a,b,c\in \F_q.$ Assume that $t=t_0(1+ut_1)$, where $t_0\in \F_q^*$ and $t_1\in \F_q.$
Then we can define a set $C_0(R)$ of length $K$ as
\begin{eqnarray*}
  C_0(R) &=& \{\frac{1}{\sqrt{K}}(\varphi(t)\lambda(t))_{t\in R^*}, \varphi\in \widehat{R}^*, \lambda\in \widehat{R}\} \\
   &=& \{\frac{1}{\sqrt{K}}(\psi(t_0)\chi_a(t_1)\chi_b(t_0)\chi_c(t_0t_1))_{t_0\in \F_q^*, t_1\in \F_q}, \psi\in \widehat{\F}_q^*,\chi_a, \chi_b, \chi_c \in \widehat{\F}_q\}.
\end{eqnarray*}
Next, we will give the two constructions of codebooks over the ring $R$.
\subsection{The first construction of codebooks}
The codebook $C_1(R)$ of length $K$ over $R$ is constructed as
\begin{eqnarray*}
  C_1(R) &=& \{\frac{1}{\sqrt{K}}(\psi(t_0)\chi_a(t_1)\chi_b(t_0)\chi_c(t_0t_1))_{t_0\in \F_q^*, t_1\in \F_q}, \\
   &&  \psi~{\rm is~a~fixed~multiplicative~character~over}~ \F_q,\chi_a, \chi_b, \chi_c \in \widehat{\F}_q\}.
\end{eqnarray*}

Based on this construction of the codebook $C_1(R)$, we have the following theorem.
\begin{thm}\label{thm2}
Let $C_1(R)$ be a codebook defined as above. Then $C_1(R)$ is a $(q^3, q(q-1))$ codebook with the maximum
cross-correlation amplitude $I_{max}(C_1(R))=\frac{1}{q-1}$.
\end{thm}
\begin{proof}
According to the definition of $C_1(R)$, it is easy to see that $C_1(R)$ has $N=q^3$ codewords of length $K=q(q-1)$. Next, our task is to determine the maximum
cross-correlation amplitude $I_{max}$ of the codebook $C_1(R)$. Let $\textbf{c}_1$ and $\textbf{c}_2$ be any two distinct codewords in $C_1(R)$, where
$\textbf{c}_1=\frac{1}{\sqrt{K}}(\psi(t_0)\chi_{a_1}(t_1)\chi_{b_1}(t_0)\chi_{c_1}(t_0t_1))_{t_0\in \F_q^*, t_1\in \F_q}$ and $\textbf{c}_2=\frac{1}{\sqrt{K}}(\psi(t_0)\chi_{a_2}(t_1)\chi_{b_2}(t_0)\chi_{c_2}(t_0t_1))_{t_0\in \F_q^*, t_1\in \F_q}$. Without loss of generality, we denote the trivial multiplicative character of $\F_q$ by $\psi_0$. Then we have
\begin{eqnarray*}
  \textbf{c}_1\textbf{c}_2^H &=&\frac{1}{K}\sum\limits_{t_0\in \F_q^*, t_1\in \F_q}\psi(t_0)\chi_{a_1}(t_1)\chi_{b_1}(t_0)\chi_{c_1}(t_0t_1)\overline{\psi(t_0)\chi_{a_2}(t_1)\chi_{b_2}(t_0)\chi_{c_2}(t_0t_1)}\\
   &=&\frac{1}{K}\sum\limits_{t_0\in \F_q^*, t_1\in \F_q}\psi_0(t_0)\chi((a_1-a_2)t_1+(b_1-b_2)t_0+(c_1-c_2)t_0t_1) \\
   &=&\frac{1}{K}\sum\limits_{t_0\in \F_q^*}\psi_0(t_0)\chi{((b_1-b_2)t_0)}\sum\limits_{t_1\in \F_q}\chi((a_1-a_2)t_1+(c_1-c_2)t_0t_1)\\
   &=& \frac{1}{K}\sum\limits_{t_0\in \F_q^*}\psi_0(t_0)\chi{(bt_0)}\sum\limits_{t_1\in \F_q}\chi((a+ct_0)t_1)~({\rm Set}~a=a_1-a_2, b=b_1-b_2, c=c_1-c_2)\\
   &=&\frac{q}{K}\sum\limits_{t_0\in \F_q^*, a+ct_0=0}\psi_0(t_0)\chi_b(t_0)\\
   &=& \frac{1}{K}G_R(\varphi, \lambda)~({\rm By~the~proof~of~Theorem~\ref{thm1},~where}~\varphi:=\psi_0\star\chi_a, \lambda:=\chi_b\star\chi_c)
\end{eqnarray*}
Since $\textbf{c}_1\neq \textbf{c}_2$, then $a, b$ and $c$ are not all equal to $0$.
In view of Theorem \ref{thm1}, we have
\begin{equation*}
K\textbf{c}_1\textbf{c}_2^H=\begin{cases}
\emph{ }-q,  ~~~~~~~~~~~{\rm if}~a=0, c=0, b\neq 0;\\
   \emph{ }q\chi(-\frac{ab}{c}), ~~~~~{\rm if}~a\neq0, c\neq0;\\
    \emph{ }0, ~~~~~~~~~~~~~~~{\rm otherwise}.\\
\end{cases}
\end{equation*}
Consequently, we infer that $|\textbf{c}_1\textbf{c}_2^H|\in \{0, \frac{1}{q-1}\}$ for any two distinct codewords $\textbf{c}_1, \textbf{c}_2$ in $C_1(R)$.
Hence, $I_{max}(C_1(R))=\frac{1}{q-1}.$
\end{proof}

By Theorem \ref{thm2}, we can calculate the ratio $\frac{I_{max}(C_1(R))}{I_w}$, which is to prove that the codebook $C_1(R)$ is asymptotically optimal.
\begin{thm}\label{thm3}
Let the symbols be the same as those in Theorem \ref{thm2}. Then the codebook $C_1(R)$ asymptotically meets the
Welch bound.
\end{thm}
\begin{proof}
In view of Theorem \ref{thm2}, note that $N=q^3$ and $K=q(q-1)$. Then the corresponding Welch bound of the codebook $C_1(R)$ is
\begin{eqnarray*}
  I_w &=& \sqrt{\frac{N-K}{(N-1)K}} \\
   &=& \sqrt{\frac{q^3-q(q-1)}{(q^3-1)q(q-1)}}\\
   &=&\sqrt{\frac{q^2-q+1}{q^4-q^3-q+1}}.
\end{eqnarray*}It follows from Theorem \ref{thm2}, then we have
\begin{equation*}
\frac{I_{max}(C_1(R))}{I_w}=\sqrt{\frac{q^4-q^3-q+1}{(q^2-q+1)(q-1)^2}}
\end{equation*}
Obviously, we get $\lim\limits_{q\longrightarrow \infty}\frac{I_{max}(C_1(R))}{I_w}=1$, which implies that $C_1(R)$ asymptotically
meets the Welch bound.
\end{proof}

\subsection{The second construction of codebooks}
The codebook $C_2(R)$ of length $K$ over $R$ is constructed as
\begin{eqnarray*}
  C_2(R) &=& \{\frac{1}{\sqrt{K}}(\psi(t_0)\chi_a(t_1)\chi_b(t_0)\chi_c(t_0t_1))_{t_0\in \F_q^*, t_1\in \F_q}, \\
   && \psi\in \widehat{\F}_q^*, \chi_b~{\rm is~a~fixed~additive~character~over}~\F_q,\chi_a, \chi_c \in \widehat{\F}_q\}.
\end{eqnarray*}

With this construction, we will figure up the maximum cross-correlation amplitude $I_{max}(C_2(R))$ as follows.
\begin{thm}\label{thm4}
Let $C_2(R)$ be a codebook defined as above. Then $C_2(R)$ is a $(q^2(q-1), q(q-1))$ codebook with the maximum
cross-correlation amplitude $I_{max}(C_2(R))=\frac{1}{q-1}$.
\end{thm}
\begin{proof}
According to the definition of $C_2(R)$, it is obvious that $C_2(R)$ has $N=q^2(q-1)$ codewords of length $K=q(q-1)$. Next, our goal is to determine the maximum
cross-correlation amplitude $I_{max}$ of the codebook $C_2(R)$. Let $\textbf{c}_1$ and $\textbf{c}_2$ be any two distinct codewords in $C_2(R)$, where
$\textbf{c}_1=\frac{1}{\sqrt{K}}(\psi_1(t_0)\chi_{a_1}(t_1)\chi_{b}(t_0)\chi_{c_1}(t_0t_1))_{t_0\in \F_q^*, t_1\in \F_q}$ and $\textbf{c}_2=\frac{1}{\sqrt{K}}(\psi_2(t_0)\chi_{a_2}(t_1)\chi_{b}(t_0)\chi_{c_2}(t_0t_1))_{t_0\in \F_q^*, t_1\in \F_q}$. Then we have
\begin{eqnarray*}
  \textbf{c}_1\textbf{c}_2^H &=&\frac{1}{K}\sum\limits_{t_0\in \F_q^*, t_1\in \F_q}\psi_1(t_0)\chi_{a_1}(t_1)\chi_{b}(t_0)\chi_{c_1}(t_0t_1)\overline{\psi_2(t_0)\chi_{a_2}(t_1)\chi_{b}(t_0)\chi_{c_2}(t_0t_1)}\\
   &=&\frac{1}{K}\sum\limits_{t_0\in \F_q^*, t_1\in \F_q}\psi_1\overline{\psi}_2(t_0)\chi((a_1-a_2)t_1+(c_1-c_2)t_0t_1)\\
   &=& \frac{1}{K}\sum\limits_{t_0\in \F_q^*}\psi(t_0)\sum\limits_{t_1\in \F_q}\chi((a+ct_0)t_1)~({\rm Set}~\psi=\psi_1\overline{\psi}_2, a=a_1-a_2, c=c_1-c_2)\\
   &=&\frac{q}{K}\sum\limits_{t_0\in \F_q^*, a+ct_0=0}\psi(t_0).
\end{eqnarray*}
\begin{itemize}
  \item If $a=c=0,$ since $\textbf{c}_1\neq \textbf{c}_2$, thus $\psi$ is nontrivial. Then we have $$K\textbf{c}_1\textbf{c}_2^H=q\sum\limits_{t_0\in \F_q^*}\psi(t_0)=0;$$
  \item If $a=0, c\neq 0$ or $a\neq0, c=0$, then $K\textbf{c}_1\textbf{c}_2^H=0$;
  \item If $a\neq0, c\neq 0$, then $K\textbf{c}_1\textbf{c}_2^H=q\psi(-\frac{a}{c})$.
\end{itemize}

\begin{equation*}
\textbf{c}_1\textbf{c}_2^H=\begin{cases}
\emph{ }\frac{q}{K}\psi(-\frac{a}{c}),  ~~~~~{\rm if}~a\neq0, c\neq0;\\
    \emph{ }0, ~~~~~~~~~~~~~~~{\rm otherwise}.\\
\end{cases}
\end{equation*}
Consequently, we infer that $|\textbf{c}_1\textbf{c}_2^H|\in \{0, \frac{1}{q-1}\}$ for any two distinct codewords $\textbf{c}_1, \textbf{c}_2$ in $C_2(R)$.
Hence, $I_{max}(C_1(R))=\frac{1}{q-1}.$
\end{proof}

Similarly, we show the near-optimality of the codebook $C_2(R)$ in the following theorem.
\begin{thm}\label{thm5}
Let the symbols be the same as those in Theorem \ref{thm4}. Then the codebook $C_2(R)$ asymptotically meets the
Welch bound.
\end{thm}
\begin{proof}
In view of Theorem \ref{thm4}, note that $N=q^2(q-1)$ and $K=q(q-1)$. Then the corresponding Welch bound of the codebook $C_2(R)$ is
\begin{eqnarray*}
  I_w &=& \sqrt{\frac{N-K}{(N-1)K}} \\
   &=& \sqrt{\frac{q^2(q-1)-q(q-1)}{(q^3-q^2-1)q(q-1)}}\\
   &=&\sqrt{\frac{q-1}{q^3-q^2-1}}.
\end{eqnarray*}It follows from Theorem \ref{thm4}, then we have
\begin{equation*}
\frac{I_{max}(C_2(R))}{I_w}=\sqrt{\frac{q^3-q^2-1}{(q-1)(q-1)^2}}
\end{equation*}
Obviously, we get $\lim\limits_{q\longrightarrow \infty}\frac{I_{max}(C_2(R))}{I_w}=1$, which implies that $C_2(R)$ asymptotically
meets the Welch bound.
\end{proof}

\section{Conclusions}
In this paper, we described the additive characters and multiplicative characters over the ring $R=\F_q+u\F_q~(u^2=0)$ in detail. Our
results on Gauss sums over the ring $R$ are calculated explicitly based on the additive and multiplicative characters. The purpose of studying the characters over $R$ is to present an application in the codebooks. Based on this idea, we proposed two constructions of codebooks and determined the maximum cross-correlation amplitude $I_{\max}(C)$ of codebooks generated by these two constructions. Moreover, we showed that these codebooks are asymptotically optimal with respect to Welch bound and the parameters of these codebooks are new.

In further research, it would be interesting to investigate the application of the new families of codebooks meeting the Welch bound or Levenstein bound by finding the new constructions of codebooks. In addition, we hope and believe that the better properties with respect to Gauss and Jacobi sums over rings will be studied and the results will be useful in applications.

\end{document}